\documentclass[12pt]{article}
\usepackage[margin=1in]{geometry}

\usepackage{amsmath,amssymb,amsthm}
\usepackage{pgfplots}
\usepackage{comment}
\usepackage[textsize=footnotesize,color=green!40]{todonotes}
\usepackage[utf8]{inputenc}
\usepackage{amsfonts}
\usepackage{amssymb}
\usepackage{mdframed}
\usepackage{framed}
\usepackage{hyperref}
\usepackage{mathtools}
\usepackage{authblk}
\usepackage{fancyhdr}
\usepackage{float}

\usepackage[noend]{algorithmic}
\usepackage[algoruled,vlined,nofillcomment,algo2e]{algorithm2e}
\usepackage{bm}

\usepackage{tikz}
\usetikzlibrary{calc,backgrounds}
\usetikzlibrary{arrows,calc,decorations.pathreplacing,positioning,shapes,shapes.geometric,shapes.callouts}
\usetikzlibrary{shapes,intersections,external}
\usepackage{pgfplots}

\newcommand{\dom}[1]{\mathbb{#1}}
\newcommand{\mech}[1]{\mathcal{#1}}
\newcommand{\tup}[1]{\vec{#1}}

\newcommand{\EE}{\mathbb{E}}
\newcommand{\PP}{\mathbb{P}}

\newcommand{\mse}{\mathrm{MSE}}
\newcommand{\SD}{\mathrm{SD}}
\newcommand{\bernoulli}{\mbox{\rm Ber}}
\newcommand{\uniform}{\texttt{Unif}}

\newcommand{\polya}{\texttt{Polya}}

\newcommand{\discretelaplace}{\texttt{DLap}}
\newcommand{\fp}{{\mbox{\rm fp}}}

\usepackage[textsize=footnotesize,color=green!40]{todonotes}

\usepackage{amsthm}

\newtheorem{theorem}{Theorem}[section]
\newtheorem{lemma}{Lemma}[section]

 \newtheorem{remark}{Remark}[section]
\date{}
\author{Borja Balle~~~James Bell\thanks{The Alan Turing Institute. Work supported by the UK Government’s Defence \& Security Programme in support of the Alan Turing Institute.}~~~Adri{\`a} Gasc{\'o}n\thanks{The Alan Turing Institute and Warwick University. Work supported by The Alan Turing Institute under the EPSRC grant EP/N510129/1, and the UK Government’s Defence \& Security Programme in support of the Alan Turing Institute.}~~~Kobbi Nissim\thanks{Department of Computer Science, Georgetown University. {\tt kobbi.nissim@georgetown.edu}. Work supported by NSF grant no.~1565387, 
TWC: Large: Collaborative: Computing Over Distributed Sensitive Data.}}
\begin{document}

\title{Differentially Private Summation with Multi-Message Shuffling}
\maketitle

\begin{abstract}
In recent work, Cheu et al. (Eurocrypt 2019) proposed a  protocol for $n$-party real summation in the shuffle model of differential privacy
with $O_{\epsilon, \delta}(1)$ error and $\Theta(\epsilon\sqrt{n})$ one-bit messages per party. In contrast, every local model protocol for real summation must incur error $\Omega(1/\sqrt{n})$, and there exist protocols matching this lower bound which require just one bit of communication per party.
Whether this gap in number of messages is necessary was left open by Cheu et al.
  
In this note we show a protocol with $O_{\epsilon, \delta}(1)$ error and $O_{\epsilon, \delta}(\log(n))$ messages of size $O(\log(n))$. This protocol is based on the work of Ishai et al.\ (FOCS 2006) showing how to implement distributed summation from secure shuffling, and the observation that this allows simulating the Laplace mechanism in the shuffle model.
\end{abstract}

\section*{A remark on concurrent work}
Concurrently and independently of our work, Ghazi et al.~\cite{DBLP:journals/corr/abs-1906-08320} obtained a similar algorithm to ours. However, there are two differences with our work that are worth stating. First, they rediscovered the algorithm from Ishai et al~\cite{ikos} to obtain secure summation from secure shuffling, and proved slightly different guarantees for it using different techniques. The second difference is in the employed distributed noise aggregation scheme: while Ghazi et al. rely on a similar technique to the one used by Shi et al.~\cite{shi}, we exploit the infinite divisibility properties of the geometric distribution, as suggested by Goryczka and Xiong~\cite{GX}. Quantitatively, these differences lead to improvements over the result of Ghazi et al.: the error of our protocol is independent of $\delta$, saving a factor of $O(\sqrt{\log(1/\delta)})$, and the communication complexity is independent of $\epsilon$ saving a factor of $O(\log(1/\epsilon))$ in the number of messages.

\section{Preliminaries}\label{sec:prelims}

\paragraph{The shuffle model.} The \emph{shuffle model} of differential privacy~\cite{erlingsson2019amplification, DBLP:journals/corr/abs-1808-01394} considers a data collector that receives messages from $n$ users (possibly multiple messages from each user). The shuffle model assumes that a mechanism is in place to provide anonymity to each of the messages, i.e.,  in the curator's view, the message have been shuffled by a random unknown permutation.

Following the notation in \cite{DBLP:journals/corr/abs-1808-01394}, we define a protocol $\mech{P}$ in the shuffle model to be a pair of algorithms $\mech{P} = (\mech{R}, \mech{A})$, where $\mech{R}: \dom{X} \to \dom{Y}^k$, and $\mech{A}: \dom{Y}^{nk} \to \dom{O}$, for number of users $n < 1$ and number of messages $k > 1$. We call $\mech{R}$ the \emph{local randomizer}, $\dom{Y}$ the \emph{message space} of the protocol, $\mech{A}$ the \emph{analyzer} of $\mech{P}$, and $\dom{O}$ the \emph{output space}.
The overall protocol implements a mechanism $\mech{P} : \dom{X}^n \to \dom{O}$ as follows.
Each user $i$ holds a data record $x_i$, to which she applies the local randomizer to obtain a vector of messages $\vec{y}_i = \mech{R}(x_i)$.
The multiset union of all $nk$ messages $y_{i,j}$ is then shuffled and submitted to the analyzer. We write $\mech{S}(\vec{y}_1,\ldots,\vec{y}_n)$ to denote the random shuffling step, where $\mech{S} : \dom{Y}^{nk} \to \dom{Y}^{nk}$ is a \emph{shuffler} that applies a random permutation to its inputs.
In summary, the output of $\mech{P}(x_1, \ldots, x_n)$ is given by $\mech{A} \circ \mech{S} \circ \mech{R}^n(\tup{x}) = \mech{A}(\mech{S}(\mech{R}(x_1), \ldots, \mech{R}(x_n)))$.

To prove privacy we will refer to the mechanism $\mech{M}_\mech{R}=\mech{S}\circ\mech{R}^n$ which captures the view of the analyzer in an execution of the protocol. Therefore we say that $\mech{P}$ is $(\epsilon,\delta)$-differentially private if for every pair of $n$-tuples of inputs $\vec{x}$ and $\vec{x}'$ differing in one co-ordinate, and every collection $T$ of multisets of $\dom{Y}$ of size $nk$, i.e. every possible subset of views of the analyzer, we have 
\begin{equation*}
  \PP(\mech{M}_{\mech{R}}(\vec{x}) \in T)\leq e^\epsilon\PP(\mech{M}_{\mech{R}}(\vec{x}') \in T)+\delta.
\end{equation*}

\paragraph{Real summation.} In this paper we are concerned with the problem of real summation where each $x_i$ is a real number in $[0,1]$ and the goal of the protocol is for the analyser to obtain a differentially private estimate of $\sum_{i=1}^n{x_i}$.

\paragraph{Randomized rounding.} Our proposed protocol uses a fixed point encoding of a real number $x$ with integer precision $p>0$ and randomized rounding, which we define
as $\fp(x, p) = \lfloor xp \rfloor + \bernoulli(xp - \lfloor xp \rfloor)$.

\begin{lemma}
For any  $\vec{x}\in\mathbb{R}^n$, $\mse(\sum_{i=1}^n \fp(x_i, p)/p, \sum_{i=1}^n x_i) \leq n/(4p^2)$.
\end{lemma}
\begin{proof}
  Let  $\Delta_i$ be $\fp(x_i, p)/p - x_i$, and note that $|\Delta_i|\leq 1/p$ and $\EE(\Delta_i)=0$. It follows that
  \begin{align*}
    \mse&\left(\sum_{i=1}^n \fp(x_i, p)/p, \sum_{i=1}^n x_i\right) = \EE\left[\left(\sum_{i=1}^n \Delta_i\right)^2\right] = \\
    &\sum_{i=1}^n \EE[\Delta_i^2] + \sum\limits_{1\leq i < j\leq n}(\EE[\Delta_i\Delta_j]) = \sum_{i=1}^n \EE[\Delta_i^2] \leq n/(4p^2).
  \end{align*}
\end{proof}

\paragraph{Differential Privacy from Statistical Distance.}
Our argument relies on statistical distance which, for consistency with~\cite{ikos}, we define as the maximal advantage of a distinguisher $A$ in telling two distributions $X$ and $Y$ apart, namely $SD(X, Y) = \max_A |\PP(A(X) = 1) - \PP(A(Y) = 1)|$. 
We will show that the view of the analyzer in our protocol is close in statistical distance to the output of a differentially private mechanism. The following lemma (also stated by Wang et al.~\cite{DBLP:conf/icml/WangFS15}, Proposition $3$) says that this suffices to conclude that our protocol is differentially private.

\begin{lemma}
\label{lemma:SD2delta}
Let $\mech{M}:\dom{I} \to \dom{O}$ and $\mech{M'}:\dom{I} \to \dom{O}$ be protocols such that
$SD(\mech{M}(i), \mech{M}'(i))\leq \mu(\sigma)$, for a security parameter $\sigma$ and all inputs $i\in\dom{I}$.
If $\mech{M}$ is $(\epsilon, \delta)$-DP, then
$\mech{M}'$ is $(\epsilon, \delta + (1+e^\epsilon)\mu(\sigma))$-DP. 
\end{lemma}
\begin{proof}
  For any neighboring inputs $i,i'\in\dom{I}$, $\mech{M}$ satisfies $\PP(\mech{M}(i)\in O) \leq e^\epsilon\PP(\mech{M}(i')\in O) + \delta$, 
  and $\mech{P'}$ satisfies $\PP(\mech{M}'(i) \in O)  \in [\PP(\mech{M}(i)\in O)-\mu(\sigma),\PP(\mech{M}(i)\in O)+\mu(\sigma)]$, for any input $i$ and
  $O\subseteq \dom{O}$. It follows that $\PP(\mech{M}'(i)\in O) \leq \PP(\mech{M}(i)\in O) + \mu(\sigma) \leq e^\epsilon\PP(\mech{M}(i')\in O) + \delta + \mu(\sigma) \leq e^\epsilon(\PP(\mech{M}'(i')\in O) + \mu(\sigma)) + \delta + \mu(\sigma)$.
\end{proof}

\paragraph{The Discrete Laplace}

In this work we use a discrete version of the Laplace mechanism, which consists of adding a discrete random variable to the input. We refer to this distribution as the {\em discrete Laplace} distribution. The distribution is over $\mathbb{Z}$, we write it $\discretelaplace(\alpha)$ and it has probability mass function proportional to $\alpha^{|x|}$. Adding noise from this distribution to a function with sensitivity $\Delta$ provides $\epsilon$-differential privacy with $\epsilon=\Delta\log(1/\alpha)$, analogously to the Laplace mechanism on $\mathbb{R}$. This distribution also appeared in \cite{shi} though under the name {\em symmetric geometric}.

\section{Secure Distributed Summation}

Ishai et al.~\cite{ikos} showed how to use anonymous communications as a building block for a variety of tasks,
including securely computing $n$-party summation over $\mathbb{Z}_q$.
This setting coincides with the shuffle model presented above, and hence the precise result by Ishai et al.\ can be restated as follows
(we give a detailed proof of this Lemma in Section~\ref{sec:IKOS}).

Let $\mech{P}$ be a shuffle model protocol, and let
$f:\dom{I}\rightarrow\dom{O}$ be a function.
We say that \emph{$P$ is $\sigma$-secure for computing $f$} if,
for any $i,j\in \dom{I}$ such that $f(i)=f(j)$, we have
$$\SD(\mech{M}_\mech{R}(i),\mech{M}_\mech{R}(j))\leq 2^{-\sigma}.$$

\begin{lemma}[\cite{ikos}]
\label{lemma:IKOS}
  There exists a $\sigma$-secure protocol $\Pi$ in the shuffle model for summation in $\mathbb{Z}_q$ with communication $O(\log(q)(\log(qn)+\sigma))$ per party.
\end{lemma}

The protocol by Ishai et al.\ is very simple. Let $x_i \in \mathbb{Z}_q$ be the input of the $i$th party.
Party $i$ generates $k=2+5\lceil\log(q)\rceil+\lceil 2\sigma+2 \log(n-1) \rceil$ additive shares of $x_i$ ($k$ can be reduced by almost a factor of two as explained in section \ref{sec:improving}), i.e., it generates $k-1$ independent uniformly random elements from $\mathbb{Z}_q$ denoted $r_{i,1},...,r_{i,k-1}$ and then computes $r_{i,k}=x_i-\sum_{j=1}^{k-1}r_{i,j}$. Party $i$ then submits each $r_{i,j}$ as a separate message to the shuffler. The shuffler then shuffles all $nk$ messages together and sends them on to the server who adds up all the received messages and finds the result $\sum_{i=1}^nx_i$ as required. This is $\sigma$-secure as stated in the lemma.

\section{Distributed Noise Addition}

Given that a communication efficient protocol for secure {\em exact} integer summation in the shuffle model exists, we would now like to use it for private real summation. Intuitively, this task boils down to defining a local randomiser that takes a private value $x_i\in [0,1]$ and outputs a privatized value $y_i$ in the \emph{discrete} domain $\mathbb{Z}_q$ such that $\sum_{i=1}^n y_i$ is differentially private and can be post-processed to a good approximation of $\sum_{i=1}^n x_i$.

A simple solution is to have a designated party add the noise required in the curator model. This is however not a satisfying solution as it does not withstand collusions and/or dropouts. To address this. Shi et al.~\cite{shi} proposed a solution where each party adds enough noise to provide $\epsilon$-differential privacy in the curator model with probability $\log(1/\delta)/n$, which results in an $(\epsilon,\delta)$-differentially private protocol. However, one can do strictly better: the total noise can be reduced by a factor of $\log(1/\delta)$ if each party adds a discrete random variable such that the sum of the contributions is exactly enough to provide $\epsilon$-differential privacy, and this also results in pure differential privacy.
A discrete random variable with this property is provided in~\cite{GX}, where it is shown that a discrete Laplace random variable can be expressed as the sum of $n$ differences of two P\'olya random variables (the P\'olya distribution is a generalization of the negative binomial distribution). Concretely, if $X_i$ and $Y_i$ are independent P\'olya$(1/n,\alpha)$ random variables then $Z=\sum_{i=1}^n X_i-Y_i$ has a discrete Laplace distribution i.e. $\PP(Z=k)\propto \alpha^{|k|}$. This allows to distribute the Laplace mechanism, which is what we shall do in our protocol presented in the next section.

\section{Private Summation}

In this section we prove a lemma which says that given a secure integer summation protocol we can construct a differentially private real summation protocol. We then combine this lemma with Lemma \ref{lemma:IKOS} to derive a protocol, given explicitly, for differentially private real summation.

\begin{lemma}
  \label{lemma:privatefromsecure}
  Given a $\sigma$-secure protocol $\Pi$ in the shuffle model for $n$-party summation in $\mathbb{Z}_q$, for any $q>0$, with communication $f(q,n,\sigma)$ per party, there exists an $(\epsilon, (1+e^\epsilon)2^{-\sigma-1})$-differentially private protocol in the shuffle model for real summation with standard error $O_\epsilon(1)$ and communication bounded by $f(\lceil 2n^{3/2} \rceil,n,\sigma)$.
\end{lemma}

\begin{proof}
  Let $\Pi$ be $(\mech{R}_\Pi, \mech{A}_\Pi,)$. We will exhibit the resulting protocol $\mech{P} = (\mech{R}, \mech{A})$, with $\mech{R} = \mech{R}_\Pi\circ \tilde{\mech{R}}$ and $\mech{A} = \tilde{\mech{A}}\circ \mech{A}_\Pi$, with $\tilde{\mech{R}}, \tilde{\mech{A}}$ defined as follows. $\mech{P}$ executes $\Pi$ with $q = \lceil 2n^{3/2} \rceil$, and thus $\tilde{\mech{R}}: [0,1] \mapsto \mathbb{Z}_{\lceil 2n^{3/2} \rceil}$. $\tilde{\mech{R}}(x_i)$ is the result of first computing a fixed-point encoding of the input $x$ with precision $p=\sqrt{n}$, and then adding noise $\eta\sim \polya(1/n,e^{-\epsilon/p})-\polya(1/n,e^{-\epsilon/p})$. $\tilde{\mech{A}}$ decodes $z$ by returning $(z - q)/p$ if $z > \frac{3np}{2}$, and $z/p$ otherwise. This addresses a potential underflow of the sum in $\mathbb{Z}_q$. To see that $\mech{P}$ has error $O(1)$, note that it has the accuracy of the discrete Laplace mechanism when adding integers, except when the total noise added has magnitude greater than $n/2$, in which case we may incur additional $O(n)$ error, but this only happens with probability $O(e^{-\epsilon n/2})$. Hence, the error of this protocol is bounded by $O(1) + O(ne^{-\epsilon n/2}) = O(1)$.

To show that this protocol is private we will compare the mechanism $\mech{M}_\mech{R}$ to another mechanism $\mech{M}_C$ (which can be considered to be computed in the curator model) which is $(\epsilon,0)$-differentially private and such that $\SD(\mech{M}_\mech{R}(\vec{x}),\mech{M}_C(\vec{x}))\leq2^{-\sigma}$ for all $\vec{x}$, from which the result follows by Lemma \ref{lemma:SD2delta}.

$\mech{M}_C(\vec{x})$ is defined to be the result of the following procedure. First apply $\tilde{\mech{R}}$ to each input $x_i$, then take the sum $s=\sum_{i=1}^n\tilde{\mech{R}}(x_i)$ and then output the result of $\mech{M}_{\mech{R}_\Pi}$ with first input $s$ and all other inputs $0$.

Note that $s = \sum_{i=1}^n\fp(x_i, p) + \discretelaplace(e^{-\epsilon/p})$, and that the sensitivity of $\sum_{i=1}^n\fp(x_i, p)$ is $p$. It follows that $s$ is $(\epsilon, 0)$-differentially private and thus by the post processing property so is $\mech{M}_C$.

It remains to show that $\SD(\mech{M}_\mech{R}(\vec{x}),\mech{M}_C(\vec{x}))\leq 2^{-\sigma}$, which we will do by demonstrating the existence of a coupling. First let the noise added to input $x_i$ by $\tilde{\mech{R}}$ be the same in both mechanisms and note that this results in the inputs to $\mech{M}_\Pi$ within $\mech{M}_\mech{R}$ and the inputs to $\mech{M}_\Pi$ within $\mech{M}_C$ having the same sum. It then follows immediately from Lemma \ref{lemma:IKOS} that these two instantiations of $\mech{M}_\Pi$ can be coupled to have identical outputs except with probability $2^{-\sigma}$, as required.
\end{proof}

The choice $p=\sqrt{n}$ was made so that the error in the discretization was the same order as the error due to the noise added, this recovers the same order of error as the curator model. Taking $p=\omega(\sqrt{n})$ results in the leading term of the total error matching the curator model at the cost of a small constant factor increase to communication.

\begin{algorithm2e}[t]
  \DontPrintSemicolon
    \SetKwInput{KwPub}{Public Parameters}
  \SetKwComment{Comment}{{\scriptsize$\triangleright$\ }}{}
\caption{Analyzer $\mech{A}_{n,k,p,q}$}\label{algo:agg}
        \KwPub{Number of parties $n$, number of messages per party $k$, precision $p$ and order $q>np$ of the additive group.}
        \KwIn{Multiset $\{y_i\}_{i\in [nk]}$, with $y _i\in \mathbb{Z}$}
        \KwOut{$z \in \mathbb{R}$}
\BlankLine
Let $z \gets \sum_{i=1}^{nk} y_i$ mod $q$ \Comment*[r]{Add all inputs mod $q$}
{\bf if }{$z>\frac{np+q}{2}$}{ \bf then }{$z \gets z-q$} \Comment*[r]{Correct for underflow}
\KwRet{$z/p$} \Comment*[r]{Rescale and return estimate}
\end{algorithm2e}

\begin{algorithm2e}[t]
\label{algo:locrand}
  \DontPrintSemicolon
  \SetKwInput{KwPub}{Public Parameters}
  \SetKwComment{Comment}{{\scriptsize$\triangleright$\ }}{}
\caption{Local Randomizer $\mech{R}_{\alpha,k,p,q}$}\label{algo:lr}
        \KwPub{Noise magnitude $\alpha$, number of messages $k$, precision $p$ and order $q>np$ of the additive group.}
        \KwIn{$x \in [0,1]$}
        \KwOut{$\vec{y} \in [0..q-1]^k$}
\BlankLine
Let $\tilde{x}\gets \lfloor xp \rfloor + \bernoulli(xp-\lfloor xp \rfloor)$ \Comment*[r]{$\tilde{x}$ is the encoding of $x$ with precision $p$}
Let $y \gets \tilde{x}+\polya(1/n,\alpha)-\polya(1/n,\alpha)$ \Comment*[r]{add noise to $\tilde{x}$}
Sample $\vec{y} \gets \uniform([ 0..q-1 ]^k)$ conditioned on $\sum_{i\in [k]}y_i=y$

\KwRet{$\vec{y}$} \Comment*[r]{Submit $k$ additive shares of $y$}
\end{algorithm2e}

Combining Lemmas \ref{lemma:IKOS} and \ref{lemma:privatefromsecure} we can conclude the following theorem.

\begin{theorem}
  There exists an $(\epsilon,\delta)$-differentially private protocol in the shuffle model for real summation with $O_{\epsilon,\delta}(1)$ error and $O_{\epsilon,\delta}(\log(n)^2)$ communication per party.
\end{theorem}

Such a protocol can be constructed from the proofs of these lemmas and is given explicitly by taking the local randomiser $\mech{R}_{\alpha,k,p,q}$ given in algorithm \ref{algo:locrand}, and the analyzer $\mech{A}_{n,k,p,q}$ given in algorithm \ref{algo:agg}, with parameters $p=\sqrt{n}$, $q=\lceil 2np \rceil$, $\alpha=e^{-\epsilon/p}$ and $k=2+5\lceil\log(q)\rceil+2\lceil \log(1/\delta)+ \log(n-1) \rceil$. This results in a mean squared error of
\begin{equation*}
  \frac{2\alpha}{(1-\alpha)^2}+\frac{n}{4p^2}
\end{equation*}
and communication of $k\lceil\log(q)\rceil$ bits per party. In section \ref{sec:improving} we explain how the choice of $k$ and thus the required communication can actually be reduced by almost a factor of two.

\section{Summation by Anonymity}
\label{sec:IKOS}

In this section we provide a proof of Lemma~\ref{lemma:IKOS}, all the ideas for the proof are provided in \cite{ikos} but we reproduce the proof here keeping track of constants to facilitate setting parameters of the protocol. The following definition and lemma from \cite{IZ89} are fundamental to why this protocol is secure.

Let $H$ be a family of functions mapping $\{0, 1\}^n$ to
$\{0, 1\}^l$ . We say $H$ is universal or a universal family of
hash functions if, for $h$ selected uniformly at random from $H$, for every $x, y \in \{0, 1\}^n$, $x\neq y$,
\begin{equation*}
  \PP(h(x) = h(y))=2^{-l}.
\end{equation*}

\begin{lemma}[Leftover Hash Lemma (special case)]
  Let $D\subset \{0, 1\}^n$, $s > 0$, $|D| \geq 2^{l+2s}$ and let $H$ be a universal family of hash functions mapping $n$ bits to $l$ bits. Let $h$, $d$ and $U$ be chosen independently uniformly at random from $H$, $D$ and $\{0,1\}^{l}$ respectively. Then 
\begin{equation*}
\SD\left( (h, h(d)), (h,U) \right)\leq 2^{-s}
\end{equation*}
\end{lemma}

To begin with we consider the case of securely adding two uniformly random inputs $X,Y\in \mathbb{Z}_q$. Recall that $\Pi$ is the protocol of the statement of the lemma, and let $V(x,y)$ be shorthand for $\mech{M}_{\mech{R}_\Pi}((x,y)) = \mech{S}\circ\mech{R}_\Pi((x,y))$, i.e. the view of the analyzer in an execution of protocol $\Pi$ with inputs $x,y$. We write $V$ for $V(X,Y)$ and $V(x)$ for $V(x,Y)$. Finally let $U$ be an independent uniformly random element of $\mathbb{Z}_q$.

\begin{lemma}
  Suppose $\log\binom{2k}{k}\geq \lceil \log(q) \rceil +2s$. Then, $\SD((V,X),(V,U))\leq 2^{-s}$.
\end{lemma}
\begin{proof}
  For $a\in \mathbb{Z}_q^{2k}$ and $\pi \in \binom{[2k]}{k}$ let $h_a(\pi)=\sum_{i\in \pi}a_i$. $(h_a)_{a\in \mathbb{Z}_q^{2k}}$ are a universal family of hash functions from $\binom{[2k]}{k}$ to $\mathbb{Z}_q$. Let $d$ be an independent uniformly random element of $\binom{[2k]}{k}$. Note that $(V,h_V(d))$ has the same distribution as $(V,X)$, which follows from the intuition that $V$ corresponds to $2k$ random numbers shuffled together, and $x$ can be obtained by adding up $k$ of them, and letting $y$ be the sum of the rest.

  The result now follows immediately from the fact that the Leftover Hash Lemma implies that $\SD((V,h_v(d)),(V,U))\leq 2^{-s}$.
\end{proof}

Now we can use this to solve the case of two arbitrary inputs.

\begin{lemma}
  If $x,y,x',y'\in \mathbb{Z}_q$ satisfy $x+y=x'+y'$, then we have
  \begin{equation*}
    \SD(V(x,y),V(x',y'))\leq 2q^2\SD((V,X),(V,U)).
  \end{equation*}
\end{lemma}
\begin{proof}
  Markov's inequality provides that
  \begin{equation*}
    \SD(V(x),V)\leq q\SD((V,X),(V,U)) \hspace{10mm} \forall x\in \mathbb{Z}_q
  \end{equation*}
  and thus by the triangle inequality
  \begin{equation*}
    \SD(V(x),V(x'))\leq 2q\SD((V,X),(V,U)).
  \end{equation*}
  Note that
  \begin{align*}
    \SD(V(x),V(x'))&=\sum_{t\in \mathbb{Z}_q}\SD(V(x)|Y=t-x,V(x')|Y=t-x')/q \\
    &=\frac{1}{q}\sum_{y\in \mathbb{Z}_q} \SD(V(x,y),V(x',y+x-x'))
  \end{align*}
  and so for every $x,y,x'\in \mathbb{Z}_q$ and $y'=y+x-x'$ we have
  \begin{equation*}
    \SD(V(x,y),V(x',y'))\leq q\SD(V(x),V(x')).
  \end{equation*}
  Combining the last two inequalities gives the result.
\end{proof}

Combining these two lemmas gives that, for $x,y,x',y'\in \mathbb{Z}_q$ such that $x+y=x'+y'$,
\begin{align}
  \SD(V(x,y),V(x',y'))&\leq 2q^22^{-\frac{\log\binom{2k}{k}-\lceil \log(q) \rceil}{2}}\nonumber \\
  \label{al:binombound}
                     &\leq 2^{-\frac{k}{2}+1+\frac{5\lceil \log(q) \rceil}{2}}
\end{align}

From which the following lemma is immediate
\begin{lemma}
  \label{lemma:twocase}
  If $k\geq 2+5\lceil\log(q)\rceil+2\sigma$ and $x,y,x',y'\in \mathbb{Z}_q$ such that $x+y=x'+y'$ then 
  \begin{equation*}
    \SD(V(x,y),V(x',y'))\leq 2^{-\sigma}
  \end{equation*}
\end{lemma}

We will now generalize to the case of $n$-party summation.

\begin{proof}[Proof of Lemma \ref{lemma:IKOS}]
Let $\vec{x},\vec{x}'\in \mathbb{Z}_q^n$ be two distinct possible inputs to the protocol, we say that they are related by a basic step if they have the same sum and only differ in two entries. It is evident that any two distinct inputs with the same sum are related by at most $n-1$ basic steps. We will show that if $k$ is taken to be $2+5\lceil\log(q)\rceil+\lceil 2\sigma+ 2\log(n-1) \rceil$ and $\vec{x}$ and $\vec{x}'$ are related by a basic step then
\begin{equation}
  \label{eq:basicbound}
  \SD(V(\vec{x}),V(\vec{x}'))\leq \frac{2^{-\sigma}}{n-1}
\end{equation}
from which the lemma follows by the triangle inequality for statistical distance.

Let $\vec{x}$ and $\vec{x}'$ be related by a basic step and suppose w.l.o.g. that $\vec{x}$ and $\vec{x}'$ differ in the first two co-ordinates. Taking $k=2+5\lceil\log(q)\rceil+\lceil 2\sigma+2 \log(n-1) \rceil$, by lemma \ref{lemma:twocase}, we can couple the values sent by the first two parties on input $\vec{x}$ with the values they send on input $\vec{x}'$ so that they match with probability $1-\frac{2^{-\sigma}}{n-1}$. Independently of that we can couple the inputs of the other $n-2$ parties so that they always match as they each have the same input in both cases. This gives a coupling exhibiting that equation \ref{eq:basicbound} holds.
\end{proof}

\begin{remark}
  It may seem counter intuitive to require more messages the more parties there are (for fixed $q$). The addition of the $\log(n-1)$ term to $k$ is necessary for the proof of Lemma \ref{lemma:IKOS}. The is because we are trying to stop the adversary from learning a greater variety of things when we have more parties. However it may be the case that Theorem 4.1 could follow from a weaker guarantee than provided by Lemma \ref{lemma:IKOS} and such a property might be true without the presence of this term.

It is an open problem to prove a lower bound greater than two on the number of messages required to get $O(1)$ error on real summation. A proof that one message is not enough is given in \cite{DBLP:journals/corr/abs-1903-02837}.
\end{remark}

\subsection{Improving the Constants}
\label{sec:improving}

  The constants implied by this proof can be improved by using a sharper bound for $\binom{2k}{k}$ in inequality \ref{al:binombound}. Using the bound $\binom{2k}{k}\geq \frac{4^k}{\sqrt{\pi(k+1/2)}}$ gives that taking $k$ to be the ceiling of the root of
\begin{equation*}
  k=1+\sigma+\frac{5\lceil \log(q)\rceil}{2}+\frac{1}{4}\log(\pi(k+\frac{1}{2}))
\end{equation*}
suffices in the statement of Lemma~\ref{lemma:twocase}. The resulting value of $k$ is
\begin{equation*}
  \frac{5}{2}\log(q)+\sigma+\frac{1}{4}\log(\log(q)+\sigma)+O(1).
\end{equation*}
Adding $\log(n-1)$ to the root before taking the ceiling gives a value of $k$ for which Lemma \ref{lemma:IKOS} holds.

\bibliographystyle{plain}
\bibliography{manymessages.bib}

\end{document}